\documentclass[12pt,letter,doublespace]{article}
\usepackage{amssymb}
\usepackage{amsmath} 
\usepackage{graphics}
\usepackage{epsfig}
\usepackage{graphicx}
\usepackage{cite}
\usepackage{amsthm}
\usepackage{color}
\usepackage{booktabs}

\newtheorem{thm}{Theorem}[section]
\newtheorem{lemma}{Lemma}[section]
\newtheorem{corol}{Corollary}[section]

\newcommand{\vx}{{\bf x}}
\newcommand{\vX}{{\bf X}}
\newcommand{\vDelta}{\mbox{\boldmath $\Delta$}}

\newcommand{\vbeta}{\mbox{\boldmath $\beta$}}

\newcommand{\vxi}{\mbox{\boldmath $\xi$}}
\newcommand{\vB}{{\bf B}}
\newcommand{\vZ}{{\bf Z}}
\newcommand{\vz}{{\bf z}}
\newcommand{\ve}{{\bf e}}
\newcommand{\vU}{{\bf U}}

\newcommand{\vY}{{\bf Y}}
\newcommand{\vP}{{\bf P}}

\newcommand{\vC}{{\bf C}}

\newcommand{\Var}{\mbox{Var}}
\newcommand{\mX}{{\mathbb X}}

\newcommand{\mW}{{\mathbb W}}

\newcommand{\pkonv}{\stackrel{p}{\rightarrow}}
\newcommand{\dkonv}{\stackrel{d}{\rightarrow}}
\newcommand{\askonv}{\stackrel{a.s.}{\longrightarrow}}

\newcommand{\bqa}{\begin{eqnarray*}}
\newcommand{\eqa}{\end{eqnarray*}}
\newcommand{\bqan}{\begin{eqnarray}}
\newcommand{\eqan}{\end{eqnarray}}
\newcommand{\bit}{\begin{itemize}}
\newcommand{\eit}{\end{itemize}}
\newcommand{\ben}{\begin{enumerate}}
\newcommand{\een}{\end{enumerate}}
\newcommand{\beq}{\begin{equation}}
\newcommand{\eeq}{\end{equation}}
\newcommand{\bdes}{\begin{description}}
\newcommand{\edes}{\end{description}}

\renewcommand{\baselinestretch}{1.7}

\begin{document}

\title{Significance Testing and Group Variable Selection}
\author{\sc Adriano Zambom and Michael Akritas\\
The Pennsylvania State University}
\date{April 27, 2012}
\maketitle{}
    \vspace{-2cm}
    
\begin{footnotetext}[1]
{Adriano Zanin Zambom: adriano.zambom@gmail.com, Michael Akritas: mga@stat.psu.edu}
\end{footnotetext}
    
\section*{Abstract}
\indent \indent Let $\vX,\ \vZ$ be $r$ and $s$-dimensional covariates, respectively, used to 
model the response variable $Y$ as $Y=m(\vX,\vZ)+\sigma(\vX,\vZ)\epsilon$. We develop 
an ANOVA-type test for the null hypothesis that $\vZ$ has no influence on the regression 
function, based on residuals obtained from local polynomial fitting of the null model. Using p-values from this test, a group variable
selection method based on multiple testing ideas is proposed.
Simulations studies suggest that the proposed test procedure outperforms the generalized likelihood ratio test when the alternative is non-additive or there is heteroscedasticity.   
Additional simulation studies, with data generated from linear, non-linear and logistic 
regression, reveal that the proposed group variable selection procedure
performs competitively against Group Lasso, and outperforms it in selecting groups having 
nonlinear effects. The proposed group variable selection procedure is illustrated on a real data set.

\vspace{.5cm}

\textbf{Keywords:} Nonparametric regression; local polynomial regression; Lack-of-fit tests;
Dimension reduction; Backward elimination.

\noindent{\bf Acknowledgments}: This research was partially supported by CAPES/Fulbright grant 15087657 and NSF grant DMS-0805598.

\newpage

\pagestyle{plain}
\setcounter{page}{1}
\setlength{\textheight}{9.0in}
\setlength{\topmargin}{-0.5in}

\section{Introduction}  \label{sec:intro}
\indent \indent Advances in data collection technologies and data storage devices have
enabled the collection of data sets involving a large number of observations on many variables in several disciplines. When the objective of data collection is that of building a predictive model for a response variable, the challenges presented by massive data sets
have opened new frontiers for statistical research. While the inclusion of a large number
of predictors reduces modeling bias, the practice of including insignificant variables
is likely to result in complicated models with less predictive power and reduced ability to 
discern and interpret the influence of the predictors. 
The underlying principles of modern 
model building are parsimony and sparseness. Parsimony requires simple models based
on few predictors. Sparseness is a relatively new concept which evolved from the 
realization that in most scientific contexts prediction can be based on only a few 
variables. Variable selection uses the assumption of 
sparseness, enabling parsimonious model building. Thus, variable (also called feature) selection plays a central role in 
current scientific research as a fundamental component of model building.

Due to readily available software, variable selection is often performed by modeling the 
expected response at covariate value
$\vx$ as $m(\vx)=\vx\vbeta$. Classical approaches to variable selection, such as 
stepwise selection or elimination 
procedures, and best subset variable selection, can be computationally intensive or 
ignore stochastic errors. A new class of methodologies addresses variable selection through 
minimization of a constrained or penalized objective function, such as Tibshirani's 
(1996) LASSO, Fan and Li's (2001) SCAD, Efron, Hastie, Johnstone and Tibshirani's 
(2004) least angle regression, Zou's (2006) adaptive LASSO, and 
Candes and Tao's (2007) Dantzig selector. A different approach exploits the conceptual 
connection between model testing and variable selection: dropping variable $j$ from the model is equivalent to not rejecting the null hypothesis $H_0^j: \beta_j=0$. Abramovich, Benjamini, Donoho and Johnstone (2006) bridged the methodological divide by showing that the application of the false discovery rate (FDR) controlling procedure of Benjamini and Hochberg (1995) on $p$-values resulting from testing each $H_0^j$ can be translated into minimizing a model selection criterion 
similar to that used in Tibshirani and Knight (1999), Birge and Massart (2001) and Foster and Stine (2004). These criteria are more flexible than that of Donoho and Johnstone (1994), which uses a penalty parameter depending only on the dimensionality of the covariate, as well as AIC and Mallow's C$_p$, which use a constant penalty parameter. Working with orthogonal designs, Abramovich et al. (2006) showed that their method is asymptotically minimax for $\ell^r$ loss, $0<r\le 2$, simultaneously throughout a range of sparsity classes, provided the level $q$ for the FDR is set to $q<0.5$. Generalizations of this methodology to non-orthogonal designs differ mainly in the generation of the $p$-values for testing $H_0^j: \beta_j=0$, and the FDR method employed. Bunea, Wegkamp and Auguste  (2006) use $p$-values generated from the standardized regression coefficients resulting from fitting the full model and employ Benjamini and Yekuteli's (2001) method for controlling FDR under dependency, while  Benjamini and Gavrilov (2009) use $p$-values from a forward selection procedure where the $i$th stage $p$-to-enter is the $i$th stage constant in the multiple-stage FDR procedure in Benjamini, Krieger and Yekutieli (2006).

Model checking and variable selection procedures based on the assumption of a linear model
may fail to discern the relevance of covariates whose effect on $m(\vx)$ is nonlinear.   Because of this, procedures for both model checking and variable selection have been developed under more general/flexible models. See, for example Li and Liang (2008), Wang and Xia (2008), Huang, Horowitz and Wei (2010), 
Storlie, Bondell, Reich and Zhang (2011), and references therein.
However, the methodological approaches in this literature have been 
distinct from those of model checking.  Working under a fully nonparametric
regression model, Zambom and Akritas (2012) developed a  competitive
variable selection procedure by exploiting the aforementioned conceptual connection 
between model checking and variable selection. Their approach consists of  backward elimination using the 
Benjamini and Yekuteli (2001) method 
applied on the $p$-values resulting from testing the significance of each covariate. The test 
procedure they developed is based on the residuals obtained by fitting all covariates except the one whose significance is being tested. These residuals serve as the response variable in a 
one-way high-dimensional ANOVA design whose factor levels are the values of the covariate being tested. By augmenting these factor levels, and using smoothness assumptions, they developed an asymptotic theory for an ANOVA-type test statistic. 

In many applications, covariates come in groups. For example, 
microarray experiments generate very large datasets with expression levels for thousands of genes but, typically, small sample size. Studies show that genes can act together as groups, and the scientific task is that of selecting the groups that are strongly associated with an outcome variable of interest. This type of problem can be addressed by 
first forming groups of genes through a clustering method and then selecting the important groups
through a group selection procedure. One of the most common group selection procedures is the Group Lasso (Yuan and Lin, 2006), and the Adaptive Group Lasso (Wang and Xia, 2008). 
See also Park, Hastie and Tibshirani (2007) who, using averages of the genes within each group, perform a selection based on a procedure combining hierarchical clustering and Lasso.

The first part of is paper develops an extension of the ANOVA test procedure of 
Zambom and Akritas (2012) to testing the significance of a group of variables under a fully nonparametric model which also allows for heteroscedasticity. The second part of the paper
introduces a backward elimination procedure for group variable selection using the 
Benjamini and Yekuteli (2001) method 
applied on the $p$-values resulting from testing the significance of each group.

This paper is organized as follows. Section \ref{sec:main} describes the proposed
methodology for testing the significance of a group of variables, derives the asymptotic null distribution of the test statistic, and
presents results of simulation studies comparing its performance to that of the generalized
likelihood ratio test of Fan and Jiang (2005). Section \ref{section:Variable_Selection} describes the test-based group variable selection procedure, and presents results of simulation studies comparing its performance to that of Group Lasso.  The analysis of a real
data set  involving gene expression levels of healthy and cancerous colon tissues 
is presented in Section \ref{sec:Real_Data}.

\section{Nonparametric Model Checking}  \label{sec:main}
\subsection{The Hypothesis and the Test Statistic}\label{test.stat}
\indent \indent Assume we have $n$ observations, $(Y_i,\vU_i)$, $i=1,\ldots,n$, of the response variable $Y$ and covariates $\vU = (\vX,\vZ)$, where $\vX$ and $\vZ$ have dimensions $r$ and $s$ respectively ($r + s = d$). Let  $m(\vx,\vz) = E(Y|\vX=\vx,\vZ=\vz)$ denote the regression function. The heterocedastic nonparametric regression model is 
\begin{equation}
Y = m(\textbf{X},\textbf{Z}) + \sigma(\textbf{X},\textbf{Z})\epsilon,
\end{equation}
 where $\epsilon$ has zero mean and constant variance and is independent from  $\vX$ and $\vZ$.
 The goal is to test the null hypothesis that $\vZ$ does not contribute to the regression function, i.e. 
\begin{equation} \label{rel.H0} 
H_0: m(\vx,\vz) = m_1(\vx).
\end{equation} 

The idea for testing this hypothesis is to treat the covariate values $\vZ_i$, 
$i=1,\ldots,n$, as the levels of high-dimensional one-way ANOVA design, with the 
null hypothesis residual $\hat\xi_i = Y_i - \hat m_1(\vx_i)$ being the observation
from factor level $\vZ_i$, and construct an ANOVA based test statistic. 
Because the asymptotic theory for high-dimensional
ANOVA requires more than one observation per factor level (Akritas and 
Papadatos, 2004), we will employ smoothness conditions, which will be 
stated below, and augment each factor level by including residuals 
from nearby covariate values. With a univariate covariate, such factor level 
augmentation was carried out in Wang, Akritas and Van Keilegom (2008) and 
Zambom and Akritas (2012) by ordering the covariate values and including 
in each factor level the residuals corresponding to neighboring covariate values. 
With a multivariate covariate, the challenge is to order the factor levels, and
hence the residuals, in a meaningful way resulting in a test statistic with 
good power properties. To do so, we propose to replace each $\vZ_i$ 
by a nonlinear version of Bair, Hastie, Paul and Tibshirani's (2004) first 
supervised principal component (PC), $P_{\theta,i}=\vZ_i^T\vC_\theta$. 
The subscript 
$\theta$ will be explained below when the supervised PC is introduced.

Having a univariate surrogate of $\vZ$, we augment each cell $P_{\theta,i}=\vZ_i^T\vC_\theta$ by including additional $p-1$, for $p$ odd, residuals $\hat\xi_\ell$ which correspond to the $p-1$ nearest neighbors $P_{\theta,\ell}$ 
of $P_{\theta,i}$.
To be specific, we consider the $(\hat\xi_i,P_{\theta,i})$, $i=1,\ldots,n$, arranged so that $P_{\theta,i_1} < P_{\theta,i_2}$ whenever 
$i_1<i_2$, and for each $P_{\theta,i}$, $(p-1)/2<i\le n-(p-1)/2$, define the nearest neighbor window $W_i$ as  
 \bqan\label{def.Wi}
W_i(\vC_{\theta}) = \left\{j:|\hat{F}_{P}(P_{\theta,j}) - \hat{F}_{P}(P_{\theta,i})| \leq 
\frac{p - 1}{2n}\right\},
 \eqan
where $\hat{F}_{P}$ is the empirical distribution function of  $P_{\theta,1},\ldots,P_{\theta,n}$. $W_i(\vC_{\theta})$ defines the augmented cell corresponding to $P_{\theta,i}$. Note that the augmented cells are defined as sets of indices rather than as sets of $\hat \xi_i$ values. 
The vector of $(n-p+1)p$ constructed "observations" in the augmented one-way ANOVA design is  
 \bqan\label{hat.vxi}
 \hat \vxi_{\vC_{\theta}} = (\hat{\xi}_j, j \in W_{(p-1)/2+1}(\vC_{\theta}), \ldots , \hat{\xi}_j, j \in W_{n-(p-1)/2}(\vC_{\theta}))^T .
 \eqan
\indent Let $\mbox{MST}$ and $\mbox{MSE}$
denote the balanced one-way ANOVA mean squares due to treatment and error, respectively, computed on the data $\hat\vxi_{\vC_{\theta}}$. The proposed test statistic is based on
 \bqan\label{rel.proposedTS}
MST-MSE.
 \eqan

In this paper the residuals $\hat\xi_i = Y_i - \hat m_1(\vx_i),\  i=1,\ldots,n,$ will 
be formed using the local polynomial of order $q$ regression estimator 
defined by
\bqan \label{local_poly_estimator}
\hat{m}_1(\vX_{i}) = \ve_1^T\left(\mX_{\vX_{i}}^T\mW_{\vX_{i}}\mX_{\vX_{i}}\right)^{-1}\mX_{\vX_{i}}^T\mW_{\vX_{i}}\vY = \sum_{j=1}^{n}\tilde w(\vX_i, \vX_j)Y_j, \mbox{      } i = 1,\ldots,n,
\eqan
where $\mW_{\vx} = diag\{K_{H_n}(\vX_{1} - \vx), \ldots, K_{H_n}
(\vX_{n} - \vx) \}$, with $K_{H_n}(\vx) = |H_n|^{-1/2}K(H_n^{-1/2}\textbf{x})$
for $K(\cdot)$ 
a bounded, non-negative $r$-variate kernel function of bounded variation 
and with bounded support and $H_n^{1/2}$ is a symmetric positive definite 
$r\times r$ bandwidth matrix, and 
\bqa
\mX_{\vx} =  \left( \begin{array}{cccc}
1      & (\vX_1 - \vx)^T & \mbox{vech}^T\left\{(\vX_1 - \vx)(\vX_1 - \vx)^T\right\} & \ldots\\
\vdots & \vdots          & \vdots & \ldots\\
1      & (\vX_n - \vx)^T & \mbox{vech}^T\left\{(\vX_n - \vx)(\vX_n - \vx)^T\right\} & \ldots \end{array} \right),
\eqa
with vech denoting the half-vectorization operator, is the $n \times \gamma_{r,q}$ design matrix, where
\bqa \label{def.gamma_d}
\gamma_{r,q} = \sum_{j=0}^q\mathop{\sum_{k_1=0}^{j}\ldots\sum_{k_r=0}^{j}}_{k_1+\ldots + k_r= j}1.
\eqa

We finish this section with a description of the construction of the first non-linearly supervised principal component $P_\theta=\vZ^T\vC_\theta$. Let $p_j, j=1,\ldots,s$, 
denote the p-values obtained by applying the test of Zambom and Akritas (2012) for testing the hypothesis $H_0^j$ which specifies that $Z_j$, the $j$th coordinate of 
$\vZ$, has no effect on the regression function of the model with response variable $Y$ and covariate vector $(\vX,Z_j)$. For a threshold parameter $\theta$, define the index
set ${\cal J}=\{j: p_j<\theta\}$ and let $\vZ_{\cal J}$ be the vector formed from the ${\cal J}$ coordinates of $\vZ$.  Then, $P_{\theta}=\vZ^T\vC_\theta$ is the first principal component of $\vZ_{\cal J}$. Note that some entries of $\vC_{\theta}$ are equal to 0, corresponding to the coordinates of $\vZ$ with $p_j$ greater or equal to $\theta$. It is important to keep in mind that the observable vector of first nonlinear principal components,
$\vP_{\theta}=(P_{\theta,1},\ldots,P_{\theta,n})$, depends on the estimated residuals $\hat \xi_i, i=1,\ldots,n$, to the extend that ${\cal J}$, and hence $\vC_\theta$, depend on them.

 \subsection{Asymptotic null distribution}
 
\begin{thm} \label{thm:main} 
Assume that the marginal densities $f_{\vX}$, 
$f_{\vZ}$ of $\vX$, $\vZ$, respectively, are bounded away from zero, the $q+1$ derivatives of $m_1(\vx)$ 
are uniformly continuous and bounded, that 
$\sigma^2(.,\vz) := E(\xi^2|\vZ^T\vC)$ is Lipschitz continuous,  $\sup_{\vx,\vz}\sigma^2(\vx,\vz)<\infty$, and
$E(\epsilon^4_i)<\infty$. 
Assume that the eigenvalues, $\lambda_i,\  i=1,\ldots,r$, of the bandwidth matrix $H_n^{1/2}$ converge to zero at the same rate and satisfy
\bqan\label{cond.lambda_poly}
n\lambda_i^{4(q+1)}\to 0\ \mbox{ and }\ \frac{n\lambda_i^{2r}}{(\log n)^2}\to \infty,\  i=1,\ldots,r. 
\eqan
Then, under $H_0$ in (\ref{rel.H0}), the asymptotic distribution of the test statistic
in (\ref{rel.proposedTS}) is given by
$$n^{1/2} (MST-MSE) \dkonv N(0,\frac{2p(2p-1)}{3(p-1)}\tau^2),$$
where $\tau = \int\left[\int\sigma^2(\vx,\vz)f_{\vX|\vZ^T\vC=\vz^T\vC}(\vx)d\vx\right]^2f_{\vZ^T\vC}(\vz^T\vC)d(\vz^T\vC)$.
\end{thm}

An estimate of $\tau^2$ can be obtained by modifying Rice's (1984) estimator as follows
\begin{equation}
\hat{\tau}^2 = \frac{1}{4(n-3)}\sum_{j=2}^{n-2}(\hat{\xi}_j-\hat{\xi}_{j-1})^2(\hat{\xi}_{j+2}-\hat{\xi}_{j+1})^2.
\end{equation}

Asymptotic theory under local additives and under general local alternatives is derived
in Zambom (2012). As these limiting results show, the asymptotic mean of the test statistic $MST - MSE$ is positive under alternatives. Thus, the test procedure rejects the null hypothesis for "large" values of the test statistic.

\subsection{Simulations: Model Checking Procedures}\label{sec:simulations:model_checking}

\indent \indent We compare the proposed ANOVA-type hypothesis test for groups with the generalized likelihood ratio test of Fan and Jiang (2005). The data is generated under three situations: a homoscedastic additive model, a homoscedastic non-additive model, and a heterocedastic non-additive model. 
All covariates, in all models, are independent standard normal.  The homoscedastic additive model is 
\bqan\label{model1}
Y = X_1 + \theta(Z_1 + Z_2 + Z_3) + \epsilon,\ \mbox{ where }\ \epsilon  \sim N(0,1),
\eqan
the homoscedastic non-additive model is
\bqan\label{model2}
Y=X_1^{X_2}(1+\theta(Z_1 + Z_2)) + X_2^{\theta(Z_1 + Z_2)}  + \epsilon,
\ \mbox{ where }\ \epsilon  \sim N(0,.1^2),
\eqan
and the heterocedastic non-additive model is
\bqan\label{model3}
Y=X_1 + \theta \sin(Z_1 Z_2) + Z_1 Z_2 \epsilon, \ \mbox{ where }\ \epsilon  \sim 
N(0,.5^2).
\eqan
In each situation we simulate 2000 data sets ot size $n = 200$. All simulations were performed in R.

In order to evaluate the effect of the threshold parameter $\theta$ we applied our test procedure with $\theta=0.05$ and $\theta=0.2$. Moreover, in each case we considered two rules to form the set of covariates from which the first supervised principal component is obtained. Rule 1 consists of using only the covariates with p-value less than $\theta$, and in Rule 2 we consider the set of covariates chosen from Rule 1 
and add to the set the covariate with the smallest p-value among the remainder covariates. In each case, if the number of  selected covariates is less than two the set is formed from the two with the smallest p-value. Thus the simulations consider four versions of our test statistic: a) Rule 1 with $\theta=0.05$, b) Rule 1 with $\theta=0.2$,
c) Rule 2 with $\theta=0.05$, d) Rule 2 with $\theta=0.2$. All four versions of our test statistic use windows of $p=11$.

Tables \ref{table:GroupLASSO_test1}, \ref{table:GroupLASSO_test2}, and
\ref{table:GroupLASSO_test3}, show the simulation results for models (\ref{model1}),(\ref{model2}), and (\ref{model3}), respectively. It is seen that the proposed test 
procedure is robust to the choice of the threshold parameter, and to the rules for selecting the set of covariates from which the first supervised principal component is obtained. The Generalized Likelihood Ratio test, which is designed for homoscedastic additive models, achieves better power under model (\ref{model1}), but is extremely liberal under heteroscedasticity and its power for the non-additive alternatives of model 
(\ref{model3}) is mainly less than its level; see Table \ref{table:GroupLASSO_test3}. Table \ref{table:GroupLASSO_test2} suggests that the GRLT has low power against non-additive alternatives even in the homoscedastic case.

\renewcommand{\baselinestretch}{1}
\begin{table}[ht]
\caption{Rejection rates for the homocedastic additive model}  \label{table:GroupLASSO_test1}
\centering     
\begin{tabular}{l c c c c c}  
\hline\hline  
& \multicolumn{5}{c}{$\theta$}  \\
  \cmidrule(r){2-6}
Method  & 0 & .2 & .4 & .6 & .8   \\ [0.5ex]  
\hline
 ANOVA-type-a  & .066 & .404 & .691 & .706 & .751  \\
 ANOVA-type-b  & .060 & .378 & .613 & .689 & .733 \\
 ANOVA-type-c  & .066 & .396 & .600 & .692 & .749\\
 ANOVA-type-d  & .057 & .375 & .618 & .685 & .718 \\
 GRLT                & .048 & .883 & 1 & 1 & 1 \\
\hline
\hline    
\end{tabular} 
\end{table}
\renewcommand{\baselinestretch}{1.7}

\renewcommand{\baselinestretch}{1}
\begin{table}[ht]
\caption{Rejection rates for the homocedastic non-additive model}  \label{table:GroupLASSO_test2}
\centering     
\begin{tabular}{l c c c c c }  
\hline\hline  
& \multicolumn{5}{c}{$\theta$}  \\
  \cmidrule(r){2-6}
Method  & 0 & .02 & .04 & .06 & .08  \\ [0.5ex]  
\hline
 ANOVA-type-a    & .051 & .202 & .522 & .693 & .724 \\
 ANOVA-type-b    & .047 & .192 & .560 & .710 & .739 \\
 ANOVA-type-c    & .050 & .193 & .520 & .679 & .711 \\
 ANOVA-type-d    & .047 & .161 & .510 & .676 & .733 \\
 GRLT                  & .052 & .059 & .117 & .235 & .379 \\
\hline
\hline    
\end{tabular} 
\end{table}
\renewcommand{\baselinestretch}{1.7}

\renewcommand{\baselinestretch}{1}
\begin{table}[ht]
\caption{Rejection rates for the heterocedastic non-additive model}  \label{table:GroupLASSO_test3}
\centering     
\begin{tabular}{l c c c c c }  
\hline\hline  
& \multicolumn{5}{c}{$\theta$}  \\
  \cmidrule(r){2-6}
Method  & 0 & .3 & .6 & 1 & 2  \\ [0.5ex]  
\hline
 ANOVA-type-a    & .035 & .168 & .503 & .654 & .789 \\
 ANOVA-type-b    & .040 & .172 & .529 & .663 & .767 \\
 ANOVA-type-c    & .037 & .161 & .501 & .651 & .757 \\
 ANOVA-type-d    & .036 & .190 & .520 & .657 & .742 \\
 GRLT                  & .584 & .585 & .535 & .439 & .297 \\
\hline
\hline    
\end{tabular} 
\end{table}
\renewcommand{\baselinestretch}{1.7}

\section{Nonparametric Group Variable Selection}\label{section:Variable_Selection}

\indent \indent In this section we will present a test-based group variable 
selection. For this purpose we will make a slight change in notation by letting
$\vX$ denote the entire vector of covariates. Thus, we consider the nonparametric regression model
\begin{equation}
Y_{i} = m(\textbf{X}_{i}) + \sigma(\vX_i) \varepsilon_{i},\  i = 1,\ldots,n,
\end{equation}
 where $\varepsilon_{i}$ is the independent error with zero mean and constant variance.  Suppose that the covariates are classified in $d$ groups 
identified by the indices $J_\ell = \{j: X_j \mbox{ belongs to group $\ell$}\}$, 
$\ell = 1, \ldots, d$, and let $s_\ell$ denote the size of group $J_\ell$. Moreover,
we assume sparseness in the sense that only the variables in a
subset $I_0=\{J_1,\ldots,J_{d_0}\}\subset \{J_1,\ldots,J_d\}$ of the groups influence the regression function. Finally, we will assume the dimension 
reduction model of Li (1991), i.e.
\bqan\label{dim.red.model}
m(\vx)=g(\vB\vx), \ \mbox{where  $\vB$ is a $K\times (\sum_i s_i)$ matrix}.
\eqan

Define the hypothesis
\bqa 
H_{0}^\ell: m(\vx) = m_1(\vx_{(-J_\ell)}), \mbox{    } \ell = 1,\ldots, d,
\eqa
where $\vx_{(-J_\ell)}$ is the set of all covariates except those whose 
index are in $J_\ell$. Under the dimension reduction model 
(\ref{dim.red.model}), this hypothesis 
can be written equivalently as
\begin{equation} \label{H0_benjamini} 
H_{0}^\ell: g(\vB\vx) = g(\vB_{(-J_\ell)}\vx_{(-J_\ell)}), \mbox{    } \ell = 1,\ldots, d,
\end{equation} 
where $\vB_{(-J_\ell)}$ is the $K\times (d-s_\ell)$ matrix obtained by omitting the columns of $\vB$ with indices in $J_\ell$. Let $\widehat\vB$ denote the Sliced Inverse Regression (SIR) estimator 
of $\vB$, and $\widehat\vB_{(-J_\ell)}$ be the corresponding submatrix. With this notation, let  
\bqa
z_\ell = \sqrt{n}(MST_\ell - MSE_\ell)/\sqrt{\frac{2p(2p-1)}{3(p-1)}}\hat\tau_\ell^4
\eqa
be the test statistic  for testing the hypothesis
(\ref{H0_benjamini}) with $\widehat\vB_{(-J_\ell)}\vX_{(-J_\ell)}$ playing the role
of $\vX$ in Theorem \ref{thm:main}, and $\widehat\vB_{(J_\ell)}\vX_{(J_\ell)}$ playing the role of $\vZ$, where $\vX_{(J_\ell)}$ is  the set of all covariates
whose index are in $J_\ell$   and $\widehat\vB_{(J_\ell)}$ is the corresponding submatrix of $\widehat\vB$.

In this context we will describe the following group variable selection procedure using backward elimination based on the Benjamini and Yekuteli (2001) method for controlling the false discovery rate (FDR):
\ben
  \item\label{ret.step1} Compute the $p$-value for $H_0^\ell$ as $\pi_{\ell} = 1-\Phi(z_\ell)$, $\ell=1,\ldots,d$.
  \item\label{ret.step2} Compute 
  \bqan\label{B-Y-k}
k = \max\left\{ i: \pi_{(\ell)} \leq \frac{\ell}{d}\frac{\alpha}{\sum_{j=1}^{d}j^{-1}}\right\}
\eqan
for a choice of level $\alpha$, where $\pi_{(1)},\ldots,\pi_{(d)}$ are the ordered p-values. If $k=d$ stop and retain all groups. If $k<d$ 
  \ben
  \item update $\vx$ by eliminating the covariates of the group corresponding to 
  $\pi_{(d)}$,  
    \item update  $d$ to $d-1$,
  \item update $\widehat\vB$ by eliminating the columns corresponding to the deleted variables, 
  \item update the test statistic  $z_\ell$, $\ell=1,\ldots,d$.
  \een 
 \item Repeat steps \ref{ret.step1} and \ref{ret.step2}, with the updated 
 $z_\ell$, $\ell=1,\ldots,d$.
\een
{\bf Remark}  Another approach for constructing a group variable selection procedure is to use a single application of the Benjamini
and Yekuteli (2001) method for controlling the false discovery rate (FDR). This is
similar to one of the two procedures proposed in Bunea et al. (2006). However, this did not perform well  in simulations and is not recommended. 
A backward elimination approach was used in Li, Cook and Nachtsheim (2005), but without incorporating multiple testing ideas.\\

\subsection{Simulations: Variable selection procedure}\label{sec.sim_var.sel}

\indent \indent In this section we compare the variable selection based on the ANOVA-type test to the Group Lasso proposed by Yuan and Lin (2006). We study the behavior of the selection for two different scenarios, one with a continuous response and another with a binary response.

For the continuous response scenario the data is generated according to the models
\bqa
\mbox{Model 1}: Y &=& X_3^3 + X_3^2 + X_3 + (1/3)X_6^3 - X_6^2 + (2/3)X_6 + \epsilon\\
\mbox{Model 2}: Y &=& sin(X_3^3 + X_3^2 + X_3) + (1/3)X_6^3 - X_6^2 + (2/3)X_6 + \epsilon\\
\mbox{Model 3}: Y &=& 10sin(X_3^3 + X_3^2 + X_3) + 5sin((1/3)X_6^3 - X_6^2 + (2/3)X_6) + \epsilon
\eqa
where $X_i = (Z_i + W)/\sqrt{2}$, $Z_i, i=1,\ldots,16$ and $W$ iid $N(0,1)$, and $\epsilon \sim N(0,2^2)$. Thus, for Models 1, 2 and 3 there are 16 groups of three covariates each, represented by the polynomial terms. The only groups that are significant are groups 3 and 6. We run 1000 simulations of data sets of size $n = 100$. Table \ref{table:GLasso1} shows the mean number of correct and incorrect groups selected by the ANOVA-type variable selection and Group Lasso using the $C_p$ criterion. It is seen that Group Lasso tends to select more groups that are not significant to the regression, while both methods perform competitively in selecting the significant groups.

\renewcommand{\baselinestretch}{1}
\begin{table} [ht]
\caption{Results for the ANOVA-type and Group Lasso}  \label{table:GLasso1}
\centering     
\begin{tabular}{l l c c }  
\hline\hline  
Model & Method & Corr.Selected & Incorr.Selected \\
\hline
Model 1 & ANOVA-type & 1.80 & .55 \\
        & Group LASSO   & 2 & 4.7 \\
\hline   
Model 2 & ANOVA-type & 1.15 & .81 \\
        & Group LASSO   & 1.59 & 4.21 \\
\hline
Model 3 & ANOVA-type & 1.84 & 0.64 \\
        & Group LASSO   & 1.80 & 6.75 \\
\hline   
\hline    
\end{tabular} 
\end{table} 
\renewcommand{\baselinestretch}{1.7}



For the second scenario, we consider the following three logistic regression models.
\bqa
\mbox{Models 1 and 2}:\ \ p_j(\vX) = \frac{1}{1 + exp(-\vbeta_j^T(1,\vX)^T)},\ j = 1, 2,
\eqa
where $\vX=(X_1,\ldots,X_{15})$ are iid $U(0,1)$, grouped sequentially in 5 groups of 3 covariates each, and
\bqa
\vbeta_1 &=& (1, -2.2, 2, 0, 0, 0, 0, 0, 0, 0, 1, 2, 0, 0, 0, 0)^T \\
\vbeta_2 &=& (1, -2.2, 3, 0, 0, 0, 0, 0, 0, 0, 1, 3, 0, 0, 0, 0)^T.
\eqa
\bqa
\mbox{Model 3}:\ \ \ p_3(\vX) = \frac{1}{1 + exp\left(-18\sin(\pi X_2) - 18\sin(\pi X_8)\right)},
\eqa
where $\vX=(X_1, \ldots,X_{12})$, with $X_1, \ldots,X_{11}$ iid $U(0,3)$ and $X_{12} \sim N(-3,1)$ independent of the others, are grouped 
sequentially in  4 groups of 3 covariates each.

The results in Table \ref{GLasso_Binary} are based on 1000 simulation runs using $n=100$ for Models 1 and 2, and $n=200$ for Model 3. It is seen that for Models 1 and 2 the number of correctly selected covariates by either procedure is low. This is probably due to the smaller sample size and the larger number of covariates. For Model 3, the Group Lasso fails to select covariates, while the ANOVA-type procedure seems to perform very well. In summary, the simulation results suggest that the ANOVA-type variable selection procedure outperforms the Group Lasso when the logistic regression model involves a non-linear function of the covariates, and has competitive performance in the other cases.

\renewcommand{\baselinestretch}{1}
\begin{table} 
\caption{Results for logistic regression}  \label{GLasso_Binary}
\centering     
\begin{tabular}{l l c c}  
\hline\hline  
Model & Method & Corr.Selected & Incorr.Selected  \\ [0.5ex]  
\hline
Model 1 & ANOVA-type & .340 & .261   \\
           & Group LASSO   & .197 & .032   \\
\hline   
Model 2 & ANOVA-type & .287 & .312   \\
           & Group LASSO   & .100 & .021  \\
\hline   

Model 3 & ANOVA-type & 1.223 & 0.080  \\
& Group LASSO   & 0.040 & 0.039   \\
\hline   
\hline    
\end{tabular} 
\end{table} 
\renewcommand{\baselinestretch}{1.7}

\section{Real Data Example} \label{sec:Real_Data}

\indent \indent The proposed procedure will be illustrated with an analysis of
the colon cancer dataset of Alon et al. (1999). The dataset was obtained from the Affymetrix technology and shows expression levels of 40 tumor and 22 
normal colon tissues of 6,500 human genes. A selection of 2,000 genes with highest minimal intensity across the samples has been made by Alon et al. (1999) and is 
publicly available at http://microarray.princeton.edu/oncology. Different clustering 
methods have been applied to this data set in several previous studies including
Dettling and Buhlman (2002, 2004),  and Ma,  Song and Huang (2007). 

To illustrate the proposed ANOVA-type group variable selection procedure we first apply 
a clustering method to form the groups. We chose the  
supervised clustering procedure {\it Wilma} proposed by Dettling and Buhlman (2002) which
is available in the package supclust in R. {\it Wilma} requires as input the number of clusters to be formed, and we specified 60, 55, 50 and 45 clusters. The next step of
the proposed procedure requires dimension reduction through SIR. 
However, because 
the number of genes (2,000) is much larger than the sample size (62) it is not 
possible to use SIR straightforward. Therefore to estimate $\vB$, we ran SIR on the set of predictors
composed of the first supervised principal component of each cluster. 

We also ran the Group Lasso procedure for binary responses using the package
grplasso in R on the same clusters/groups returned by Wilma. However, a corresponding modification is needed for the calculation of the degrees of freedom needed for the application of 
the $C_p$ criterion;
see Yuan and Lin (2006). This calculation requires the estimator $\hat \beta$ from 
fitting all individual covariates. 
Since the number of covariates is much larger than the sample size, we obtain an approximation 
to the required estimator by first obtaining the estimator $\hat\beta_P$ from
fitting the first PC from each cluster. Since the PCs are linear combinations of 
the covariates in each cluster, having a coefficient for a cluster's PC translates into coefficients
for the covariates in that cluster.

\renewcommand{\baselinestretch}{1}
\begin{table} [htt]
\caption{Results for Colon data set}  \label{Colon}
\centering     
\begin{tabular}{l l c }  
\hline\hline  
No. Initial Clusters & Procedure & Clusters Selected\\
60 & ANOVA-type & 1, 2, 3, 4, 5, 6, 8, 9, 10, 11, 12, 15, 16,\\
& & 18, 23, 24, 25, 28, 31, 33, 34, 39, 42, 46 \\
   & Group Lasso & 1, 2, 3, 8, 9, 16, 17, 37 \\
\hline   
55 &    ANOVA-type &  1, 2, 3, 4, 5, 6, 8, 10, 11, 12, 15, 16,\\
& & 17, 22, 23, 24, 25, 26, 33, 38, 42, 45 \\
   & Group Lasso &  2, 3, 7, 9, 13, 18\\
\hline   
50 & ANOVA-type & 1, 2, 3, 4, 5, 6, 8, 10, 11, 12, 14, 15, 16, \\
    & &18, 21, 23, 24, 25, 26, 36, 40, 45, 47, 48, 49\\
   & Group Lasso & 2, 3, 7, 9, 27, 29\\
\hline   
45 & ANOVA-type & 1, 2, 3, 4, 5, 6, 8, 10, 11, 12, 15, 16,\\
& & 18, 21, 22, 23, 24, 25, 31, 34, 37, 40, 42, 44, 45\\
   & Group Lasso & 2, 3, 7, 9, 10, 16\\
\hline   
\hline    
\end{tabular} 
\end{table} 
\renewcommand{\baselinestretch}{1.7}

Table \ref{Colon} shows the groups selected by the proposed group variable selection and the group lasso procedure for the different specified number of clusters returned by Wilma.

We note that there is significant overlap in the genes included in the clusters selected
by each method across the different numbers of total clusters specified. Thus, the different
number of total clusters specified is not critical for selecting the important genes. The proposed method selects more clusters, which is contrary to the simulation results. This is probably due to the linear link function for the logit used in grplasso. For example, when the
logit of the fitted probability of cancerous tissue is plotted against the first PC of cluster 4, 
which is selected by the proposed method but not by Group Lasso, it shows a nonlinear effect (left panel of Figure \ref{fig.pc}),
whereas the non-linearity is much less pronounced when plotted against the first PC of 
cluster 2, which is selected by the proposed method and by Group Lasso (right panel of 
Figure \ref{fig.pc}). 

\begin{figure}
\centering

\begin{tabular}{cc}
 \includegraphics[width=0.47\textwidth]{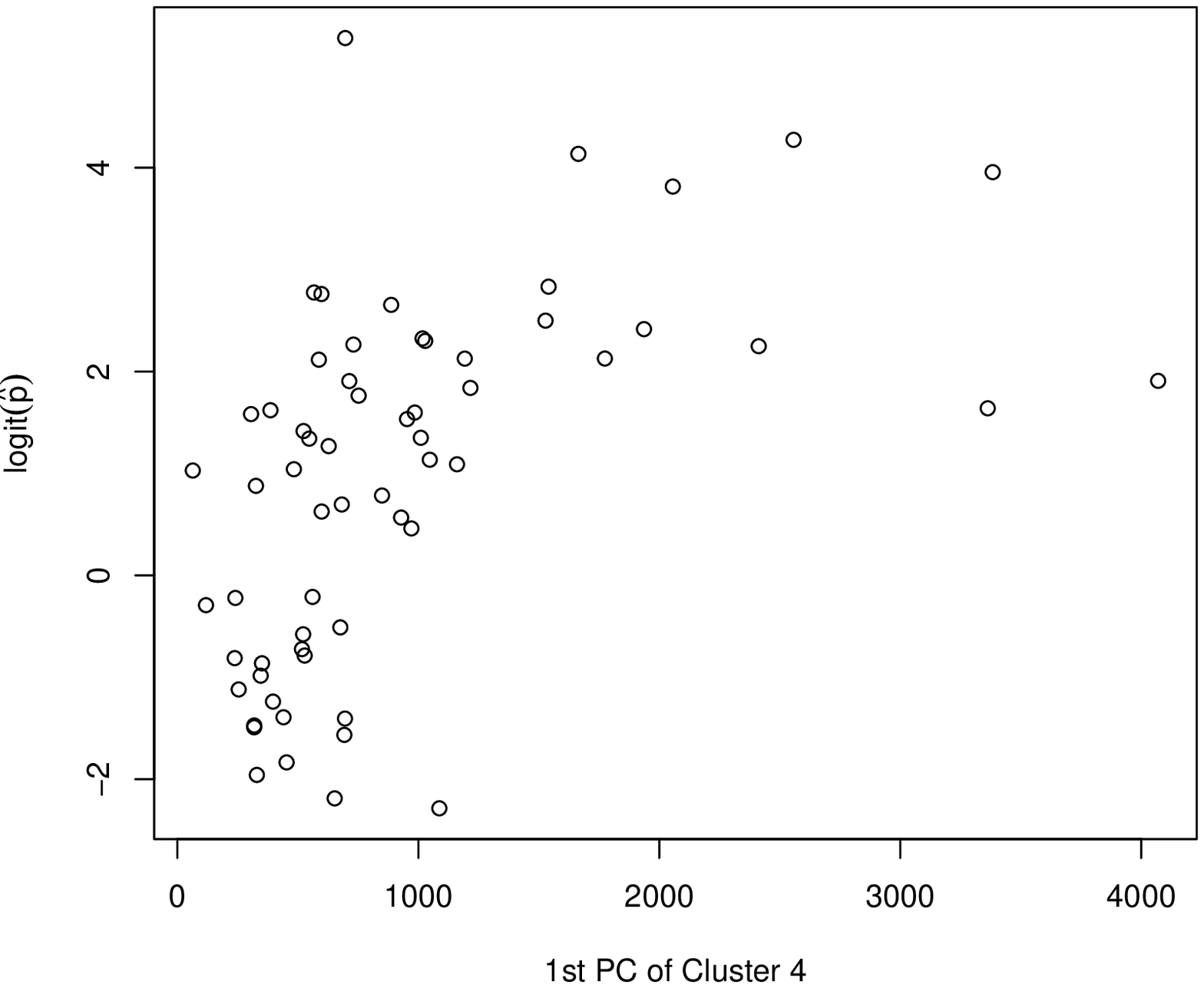} &
\includegraphics[width=0.47\textwidth]{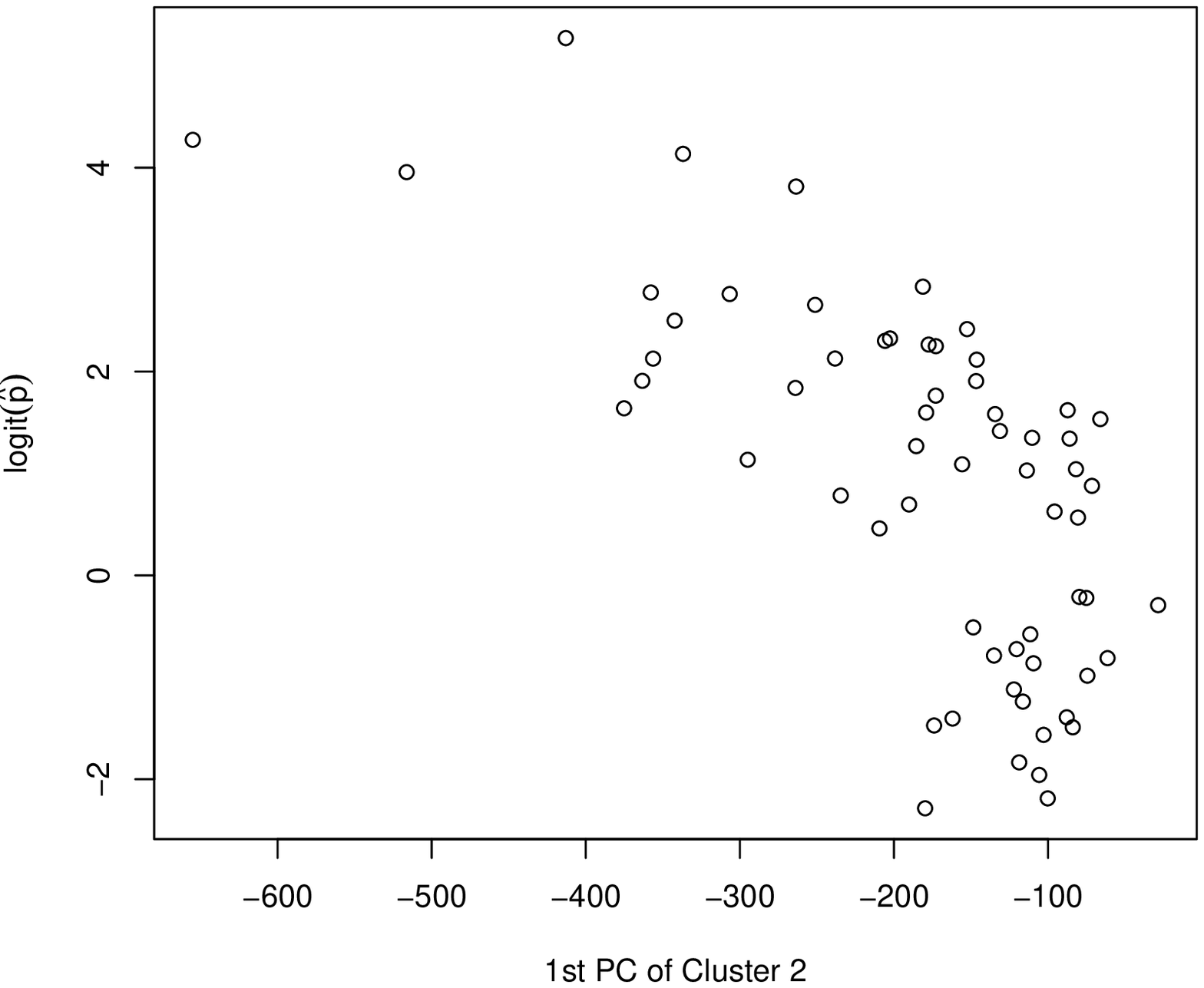}\\
\end{tabular}
\caption{Plot of Group Lasso Logit Estimate Against the First PC of Clusters 4 and 2}
\label{fig.pc}
\end{figure}

\appendix

\section{Appendix}

\begin{proof}[\textbf{Proof of Theorem \ref{thm:main}}]
\indent Under $H_0$ in (\ref{rel.H0}) we write
 \bqa
 \hat{\xi}_i &=& Y_i - \hat{m}_1(\vX_i) + m_1(\vX_i) - m_1(\vX_i) = \xi_i - (\hat{m}_1(\vX_i) - m_1(\vX_i)) \\
 &=& \xi_i - \Delta_{m_1}(\vX_i),
 \eqa
  where $\Delta_{m_1}(\vX_i)$ is defined implicitly in the above relation. Thus, $\hat\vxi_{\vC_{\theta}}$ of relation (\ref{hat.vxi}) is decomposed as
$
  \hat\vxi_{\vC_{\theta}}=\vxi_{\vC_{\theta}} - \vDelta_{m_1\vC_{\theta}},
$
  where $\vxi_{\vC_{\theta}}$ and $\vDelta_{m_1\vC_{\theta}}$ are defined as in (\ref{hat.vxi}) but using 
  $\xi_i$ and $\Delta_{m_1}(\vX_i)$, respectively, instead of $\hat\xi_i$. Note that MST-MSE given in (\ref{rel.proposedTS}) can be written as a quadratic form $\hat\vxi_{\vC_{\theta}}^TA\hat\vxi_{\vC_{\theta}}$ (see  Wang, Akritas and Van Keilegom, 2008), where
\bqan\label{def.matrixA}
A = \frac{np-1}{n(n-1)p(p-1)}\oplus_{i=1}^{n}\bold{J}_{p}-\frac{1}{n(n-1)p}\bold{J}_{np}-\frac{1}{n(p-1)}\bold{I}_{np},
\eqan
$\bold{I}_d$ is a identity matrix of dimension d, $\bold{J}_d$ is a dxd matrix of 1's and $\oplus$ is the Kronecker sum or direct sum. Thus, we can write $\sqrt{n}(\mbox{MST - MSE})$ as 
  \bqan\label{rel.decom.T}
 \sqrt{n} \hat \vxi_{\vC_{\theta}}^TA\hat\vxi_{\vC_{\theta}} = \sqrt{n}\vxi_{\vC_{\theta}}^TA\vxi_{\vC_{\theta}} -\sqrt{n} 2\vxi_{\vC_{\theta}}^TA\vDelta_{m_1\vC_{\theta}} +\sqrt{n} \vDelta_{m_1\vC_{\theta}}^TA\vDelta_{m_1\vC_{\theta}}.
  \eqan

That $\sqrt{n} 2\vxi_{\vC_{\theta}}^TA\vDelta_{m_1\vC_{\theta}}$ and $\sqrt{n} \vDelta_{m_1\vC_{\theta}}^TA\vDelta_{m_1\vC_{\theta}}$ converge in probability to 0 uniformily follows from arguments similar to those used in Zambom and Akritas (2012).

Using Corolary \ref{corol:normal}, to show the asymptotic normality of $\sqrt{n}\vxi_{\vC_{\theta}}^TA\vxi_{\vC_{\theta}}$, it is enough to show that
\bqa
\sup_{\vC}\left|P\left(\frac{\sqrt{n}\vxi_{\vC}^TA_d\vxi_{\vC}}{\frac{2p(2p-1)}{3(p-1)}\tau^2} \leq t\right) - \Phi\left(t\right)\right| \rightarrow 0.
\eqa
Let $b_n \sim n^{2/3}$ and $r_n \sim n/b_n \sim n^{1/3}$ and write
\bqan
\sqrt{n}\vxi_{\vC_{\theta}}^TA_d\vxi_{\vC_{\theta}} &=& \frac{1}{\sqrt{n}}\sum_{i=1}^n\frac{1}{p-1}\sum_{j_1\neq j_2}\xi_{j_1}\xi_{j_2}I\left(j_1,j_2 \in W_i(\vC_{\theta})\right)\nonumber\\
&=&  \frac{1}{\sqrt{n}}\sum_{i=1}^{r_n}U_i(\vC_{\theta}) + \frac{1}{\sqrt{n}}\sum_{i=1}^{r_n}V_i(\vC_{\theta})\nonumber\\
&=& \frac{1}{\sqrt{n}}S_U(C_{\theta}) + \frac{1}{\sqrt{n}}S_V(C_{\theta}), \label{SuSv}
\eqan
where, with $\gamma_i(\vC_{\theta}) = \frac{1}{p-1}\sum_{j_1\neq j_2}\xi_{j_1}\xi_{j_2}I\left(j_1,j_2 \in W_i(\vC_{\theta})\right)$, 
\bqa
U_i(\vC_{\theta}) &=& \gamma_{(i-1)(b_n+p)+1}(\vC_{\theta}) + \ldots +  \gamma_{(i-1)(b_n+p)+b_n}(\vC_{\theta}),\\
V_i(\vC_{\theta}) &=& \gamma_{(i-1)(b_n+p)+b_n+1}(\vC_{\theta}) + \ldots +  \gamma_{i(b_n+p)}(\vC_{\theta}).
\eqa
Note that the $U_i(\vC_{\theta})$ are independent, and the $V_i(\vC_{\theta})$ are independent.

Now, letting $\mbox{sd} = \sqrt{\frac{2p(2p-1)}{3(p-1)}\tau^2}$, we have 
\bqan
&&  \sup_{\vC}\left|P\left(\frac{\sqrt{n}\vxi_{\vC}^TA_d\vxi_{\vC}}{\mbox{sd}} \leq t\right) - \Phi\left(t\right)\right|\nonumber\\
& = & \sup_{\vC}\left|P\left(\frac{S_U(\vC) + S_V(\vC)}{\sqrt{n}\mbox{sd}} \leq t\right) - \Phi\left(t\right)\right|\nonumber\\
& = & \sup_{\vC}\Bigg|P\left(\frac{S_U(\vC)}{\sqrt{n}\mbox{sd}} \leq t - \frac{S_V(\vC)}{\sqrt{n}\mbox{sd}}, \left|\frac{S_V(\vC)}{\sqrt{n}\mbox{sd}}\right| \leq \epsilon \right)\nonumber\\
&&\ \ \ \ \ \ \ \ \ \ \ \ \ \ \ \ \ + P\left(\frac{S_U(\vC)}{\sqrt{n}\mbox{sd}} \leq t - \frac{S_V(\vC)}{\sqrt{n}\mbox{sd}}, \left|\frac{S_V(\vC)}{\sqrt{n}\mbox{sd}}\right| \geq \epsilon \right) - \Phi\left(t\right)\Bigg|\nonumber\\
& \leq & \sup_{\vC}\left|P\left(\frac{S_U(\vC)}{\sqrt{n}\mbox{sd}} \leq t - \frac{S_V(\vC)}{\sqrt{n}\mbox{sd}}, \left|\frac{S_V(\vC)}{\sqrt{n}\mbox{sd}}\right| \leq \epsilon \right) - \Phi\left(t\right)\right| \nonumber\\
&& + \sup_{\vC}P\left(\left|\frac{S_V(\vC)}{\sqrt{n}\mbox{sd}}\right| \geq \epsilon \right) \label{Sn_to_0}
\eqan 
That the second in (\ref{Sn_to_0}) term converges to zero follows from Lemma \ref{lemma:Sv}. That the first term in (\ref{Sn_to_0}) converges to zero follows from Lemma \ref{lemma:Su}, provided we show that 
\bqan \label{Var_Sv}
\Var\left(\frac{S_U(\vC)}{\sqrt{n}}\right) \to \mbox{sd}^2,\ \ \mbox{for any}\ \ \vC. 
\eqan
By (\ref{SuSv}), and because $\frac{S_V(\vC)}{\sqrt{n}\mbox{sd}} \pkonv 0$, (\ref{Var_Sv}) follows from $\sup_{\vC}\Var(\sqrt{n}\vxi_{\\vC}^TA_d\vxi_{\\vC}) \to \mbox{sd}^2$. By the definition of $\vxi_{\vC_{\theta}}^TA_d\vxi_{\vC_{\theta}}$, it is easy to see that
$E\left(\vxi_{\vC}^TA_d\vxi_{\vC}\right) = 0$ for any $\vC$.
To find the variance of $\sqrt{n}\vxi_{\\vC}^TA_d\vxi_{\\vC}$ we first evaluate the conditional second moment 
$E[(\sqrt{n}\vxi_{\\vC}^TA_d\vxi_{\\vC})^2| \vZ^T\vC]$. 
Recalling the notation
$\sigma^2(.,\vz_j^T\vC) = E(\xi_j^2|\vZ^T\vC=\vz_j^T\vC)$, we have
\bqan\label{rel.Pf.Thm1.bounded}
&& \sup_{\vC}\frac{1}{n(p-1)^2}\sum_{i_1,i_2}^{n}\sum_{j_1 \neq l_1}^{n}\sum_{j_2 \neq l_2}^{n}E(\xi_{j_1}\xi_{l_1}\xi_{j_2}\xi_{l_2}|\vZ^T\vC)I(j_s \in W_{i_s}(\vC),l_s \in W_{i_s}(\vC),s=1,2)
\nonumber\\[2mm]
&& = \sup_{\vC}\frac{2}{n(p-1)^2}\sum_{i_1=1}^{n}\sum_{i_2=1}^{n}\sum_{j \neq l}^{n}\sigma^2(.,\vz_{j}^T\vC)\sigma^2(.,\vz_{l}^T\vC)I(j,l \in W_{i_1}(\vC)\cap W_{i_2}(\vC))  \nonumber \\
&& = \sup_{\vC}\frac{2}{n(p-1)^2}\sum_{i_1=1}^{n}\sum_{i_2=1}^{n}\sum_{j \neq l}^{n}\sigma^2(.,\vz_{j}^T\vC)\left(\sigma^2(.,\vz_{j}^T\vC) + O_p\left(\frac{p}{\sqrt{n}}\right)\right)I(j,l \in W_{i_1}(\vC)\cap W_{i_2}(\vC)) \nonumber \\
&& = \sup_{\vC}\frac{2}{n(p-1)^2}\sum_{j=1}^{n}\sigma^4(.,\vz_{j}^T\vC)\sum_{i_1=1}^{n}\sum_{i_2=1}^{n}\sum_{l \neq j}^{n}I(j,l \in W_{i_1}(\vC)\cap W_{i_2}(\vC)) + O_p\left(\frac{p^2}{n^{1/2}}\right) \nonumber \\
&& = \sup_{\vC}\frac{2}{n(p-1)^2}\sum_{j=1}^{n}\sigma^4(.,\vz_{j}^T\vC)2(1+2^2+3^2+...+(p-1)^2) + O_p\left(\frac{p^2}{n^{1/2}}\right) \nonumber \\
&& = \sup_{\vC}\frac{2}{n(p-1)^2}\frac{p(p-1)(2p-1)}{3}\sum_{j=1}^{n}\sigma^4(.,\vz_{j}^T\vC) + O_p\left(\frac{p^2}{n^{1/2}}\right), \nonumber
\eqan
where the third equality follows from Lemma \ref{lemma:Op} using the assumption that $\sigma^2(.,\vz_j^T\vC)$ is Lipschitz continuous and the second last inequality results from the fact that if $1 \leq |j_1-j_2| = s \leq p-1$, then they are $(p-s)^2$ pairs of windows whose intersection includes $j_1$ and $j_2$. Taking limits as $n\to\infty$ it is seen that
 \vspace{-1mm}
\begin{eqnarray}
\sup_{\vC}E\left(n^{1/2}\vxi_{\vC}^TA_d\vxi_{\vC}\right|\vZ^T\vC)^2 &\askonv & \frac{2(2p-1)}{3(p-1)}E(\sigma^4(.,\vz^T\vC)) 
= \frac{2(2p-1)}{3(p-1)}\tau^2.
\end{eqnarray}
From relation (\ref{rel.Pf.Thm1.bounded}) it is easily seen that 
$\sup_{\vC}E[(\sqrt{n}\vxi_{\vC}^TA_d\vxi_{\vC})^2| \vZ^T\vC]$ remains bounded, and thus
$\sup_{\vC}\Var(n^{1/2}\xi_{\vC}^TA\xi_{\vC})$ also converges to the same limit by the Dominated Convergence Theorem.
\end{proof}

\begin{lemma} \label{lemma:Ad} 
If the assumptions of Theorem \ref{thm:main} hold, then under $H_0$ and as n $\rightarrow \infty$,
\bqan\label{lemma:Ad_eqn} 
\sup_{\vC}P\left(n^{1/2}\left|\vxi_{\vC_{\theta}}^TA\vxi_{\vC_{\theta}} - \vxi_{\vC_{\theta}}^TA_d\vxi_{\vC_{\theta}}\right| \geq \epsilon\right) \rightarrow 0,
\eqan
where $A_d = diag\{B_1,...,B_n\}$, with $B_i = \frac{1}{n(p-1)}[\bold{J}_p - \bold{I}_p].$
\end{lemma}
\begin{proof}
By Chebyshev Inequality, we have that
\bqan
\sup_{\vC}P\left(n^{1/2}\left|\vxi_{\vC}^TA\vxi_{\vC} - \vxi_{\vC}^TA_d\vxi_{\vC}\right| \geq \epsilon\right) \leq \sup_{\vC}\frac{nE\left[\left(\vxi_{\vC}^TA\vxi_{\vC} - \vxi_{\vC}^TA_d\vxi_{\vC}\right)^2\right]}{\epsilon^2} \label{Cheby}
\eqan

Since the block diagonal elements of $A_d$ equal those of $A$, it sufices to show that the off diagonal blocks of $A$ are negligible. For $i_1 \neq i_2$, every element of the block $(i_1, i_2)$ equals $\frac{1}{n(n-1)p}$. We will show that the second moment on the right hand side of (\ref{Cheby}) conditionally on $\vZ$ goes to zero, and therefore the unconditional second moment also does. To that end, write
\bqan \label{lemma:Ad_eqn}
&&\sup_{\vC}\frac{nE\left[\left(\vxi_{\vC}^TA\vxi_{\vC} - \vxi_{\vC}^TA_d\vxi_{\vC}\right)^2|\vZ\right]}{\epsilon^2}\\
&&=n\left(\frac{1}{n(n-1)p}\right)^2\sup_{\vC}E\left(\sum_{i_1\neq i_2}\sum_{i_3\neq i_4}\sum_{j_1,j_2,j_3,j_4=1}^n\xi_{j_1}\xi_{j_2}\xi_{j_3}\xi_{j_4}I(j_k \in W_{i_k}(\vC), k=1,\ldots,4)|\vZ\right) \nonumber \\
&&=n\left(\frac{1}{n(n-1)p}\right)^2\sup_{\vC}\sum_{i_1\neq i_2}\sum_{i_3\neq i_4}\sum_{j_1,j_2,j_3,j_4=1}^nE\left(\xi_{j_1}\xi_{j_2}\xi_{j_3}\xi_{j_4}|\vZ\right)I(j_k \in W_{i_k}(\vC), k=1,\ldots,4) \nonumber 
\eqan
The expected value in this sum is different from zero, only if $\xi_{j_1} , \ldots , \xi_{j_4}$ consists of two pairs of equal observations, or $j_1 = j_2 = j_3 = j_4$. Since there are $O(n^2p^4)$ terms for the former case to happen and $O(n p^4)$ for the latter case to happen, and the magnitude of these terms is not affected by $\vC$, the 
order of (\ref{lemma:Ad_eqn}) is $O\left(\frac{n}{p}\frac{1}{n^4p^2}n^2p^4\right) = o(1)$, and this completes the proof. 
\end{proof}

\begin{corol} \label{corol:normal} Let $A_d = diag\{B_1,...,B_n\}$, with $B_i = \frac{1}{n(p-1)}[\bold{J}_p - \bold{I}_p]$, $\mbox{sd} = \sqrt{\frac{2p(2p-1)}{3(p-1)}\tau^2}$, and $\vxi_\vC$ be defined in (\ref{hat.vxi}) with $\vC$ instead of $\vC_{\theta}$. Then, under the assumptions of Theorem \ref{thm:main}  we have
\bqa
&&\sup_{\vC}\sup_t\left|P\left(\frac{n^{1/2}\vxi_{\vC}^TA\vxi_{\vC}}{sd} \leq t\right) - \Phi(t)\right| \rightarrow 0 \mbox{  if and only if}\\
&&\sup_{\vC}\sup_t\left|P\left(\frac{n^{1/2}\vxi_{\vC}^TA_d\vxi_{\vC}}{sd} \leq t\right) - \Phi(t)\right| \rightarrow 0.
\eqa
\end{corol}
\begin{proof}
Write 
\bqa
\frac{n^{1/2}\vxi_{\vC}^TA\vxi_{\vC}}{sd} = \frac{n^{1/2}\vxi_{\vC}^TA_d\vxi_{\vC}}{sd} + \frac{n^{1/2}\left(\vxi_{\vC}^TA\vxi_{\vC}- \vxi_{\vC}^TA_d\vxi_{\vC}\right)}{sd}.
\eqa
Now, for any $t$
\bqan
&&\hspace{-9mm}\sup_{\vC}\left|P\left(\frac{n^{1/2}\vxi_{\vC}^TA\vxi_{\vC}}{sd} \leq t\right) - \Phi(t)\right|\nonumber\\
&&\hspace{-9mm}=\sup_{\vC}\Bigg|P\left(\frac{n^{1/2}\vxi_{\vC}^TA_d\vxi_{\vC}}{sd} \leq t - \frac{n^{1/2}\left(\vxi_{\vC}^TA\vxi_{\vC}- \vxi_{\vC}^TA_d\vxi_{\vC}\right)}{sd}, \left|\frac{n^{1/2}\left(\vxi_{\vC}^TA\vxi_{\vC}- \vxi_{\vC}^TA_d\vxi_{\vC}\right)}{sd}\right| \leq \epsilon\right)\nonumber\\
&&\hspace{-9mm} + P\left(\frac{n^{1/2}\vxi_{\vC}^TA_d\vxi_{\vC}}{sd} \leq t - \frac{n^{1/2}\left(\vxi_{\vC}^TA\vxi_{\vC}- \vxi_{\vC}^TA_d\vxi_{\vC}\right)}{sd}, \left|\frac{n^{1/2}\left(\vxi_{\vC}^TA\vxi_{\vC}- \vxi_{\vC}^TA_d\vxi_{\vC}\right)}{sd}\right| \geq \epsilon\right) - \Phi(t)\Bigg|\nonumber\\
&& \hspace{-9mm} \leq \sup_{\vC}\max\Bigg\{\Bigg|P\left(\frac{n^{1/2}\vxi_{\vC}^TA_d\vxi_{\vC}}{sd} \leq t + \epsilon \right)  - \Phi(t)\Bigg|, \Bigg|P\left(\frac{n^{1/2}\vxi_{\vC}^TA_d\vxi_{\vC}}{sd} \leq t - \epsilon \right)  - \Phi(t)\Bigg|\Bigg\}\nonumber\\
&& + \sup_{\vC}P\left( \left|\frac{n^{1/2}\left(\vxi_{\vC}^TA_d\vxi_{\vC}- \vxi_{\vC}^TA\vxi_{\vC}\right)}{sd}\right| \geq \epsilon\right). \label{to_normal}
\eqan
The last term in (\ref{to_normal}) goes to zero  by Lemma \ref{lemma:Ad}. Thus,
\bqa
&&\hspace{-9mm}\sup_{\vC}\sup_t\left|P\left(\frac{n^{1/2}\vxi_{\vC}^TA\vxi_{\vC}}{sd} \leq t\right) - \Phi(t)\right| \nonumber\\
&& \hspace{-9mm} \leq \sup_{\vC}\max\Bigg\{\sup_t\Bigg|P\left(\frac{n^{1/2}\vxi_{\vC}^TA_d\vxi_{\vC}}{sd} \leq t \right)  - \Phi(t- \epsilon)\Bigg|, \sup_t\Bigg|P\left(\frac{n^{1/2}\vxi_{\vC}^TA_d\vxi_{\vC}}{sd} \leq t  \right)  - \Phi(t+ \epsilon)\Bigg|\Bigg\}\\
&& \hspace{5mm} + o(1)\nonumber\\
&& \hspace{-9mm} \leq \sup_{\vC}\sup_t\Bigg|P\left(\frac{n^{1/2}\vxi_{\vC}^TA_d\vxi_{\vC}}{sd} \leq t \right)  - \Phi(t)\Bigg| + \sup_{t}|\Phi(t) - \Phi(t+\epsilon)| + o(1).\nonumber
\eqa
Letting $\epsilon \to 0$, 
\bqa
 \lim_{n\to \infty}\sup_{\vC}\sup_t\Bigg|P\left(\frac{n^{1/2}\vxi_{\vC}^TA\vxi_{\vC}}{sd} \leq t \right)  - \Phi(t)\Bigg| \leq \lim_{n\to \infty}\sup_{\vC}\sup_t\Bigg|P\left(\frac{n^{1/2}\vxi_{\vC}^TA_d\vxi_{\vC}}{sd} \leq t \right)  - \Phi(t)\Bigg|.
\eqa
Using similar steps, it can be shown that
\bqa
 \lim_{n\to \infty}\sup_{\vC}\sup_t\Bigg|P\left(\frac{n^{1/2}\vxi_{\vC}^TA_d\vxi_{\vC}}{sd} \leq t \right)  - \Phi(t)\Bigg| \leq \lim_{n\to \infty}\sup_{\vC}\sup_t\Bigg|P\left(\frac{n^{1/2}\vxi_{\vC}^TA\vxi_{\vC}}{sd} \leq t \right)  - \Phi(t)\Bigg|,
\eqa
completing the proof.
\end{proof}

\begin{lemma} \label{lemma:Sv} 
Let $S_V(\vC)$ be defined as in (\ref{SuSv}). Under the assumptions of Theorem \ref{thm:main},
\bqa
\sup_{\vC}P\left(\left|\frac{S_V(\vC)}{\sqrt{n}\mbox{ sd}}\right| \geq \epsilon \right) \rightarrow 0
\eqa
\end{lemma}
\begin{proof}
For any $\epsilon > 0$, since $V_i(\vC)$ are independent,
\bqa
&& \sup_{\vC}P\left(n^{-1/2}|\sum_{i=1}^{r_n}V_i(\vC)| \geq \epsilon\right) \leq \sup_{\vC}\sum_{i=1}^{r_n}P\left(|V_i(\vC)| \geq \epsilon n^{1/2}r_n^{-1}\right) \\
&& \leq \sup_{\vC} \sum_{i=1}^{r_n}\frac{E(V_i(\vC)^4)}{\epsilon^4n^2r_n^{-4}}  \leq K\epsilon^{-4}n^{-2}r_n^5(p^2)^2 = o(1),
\eqa
where the last inequality follows from the fact that $E(V_i^4(\vC)) \leq K(p^2)^2$.
\end{proof}

\begin{lemma} \label{lemma:Su} 
Let $S_U(\vC)$ and $S_V(\vC)$ be defined as in (\ref{SuSv}). Under the assumptions of Theorem \ref{thm:main},
\bqan  \label{eqn:Su} 
\sup_{\vC}\left|P\left(\frac{S_U(\vC)}{\sqrt{n}\mbox{sd}} \leq t - \frac{S_V(\vC)}{\sqrt{n}\mbox{sd}}, \left|\frac{S_V(\vC)}{\sqrt{n}\mbox{sd}}\right| \leq \epsilon \right) - \Phi\left(t\right)\right| \rightarrow 0.
\eqan
\end{lemma}
\begin{proof}
Note that, using the Berry Essen bound (see Shorack (Probability for Statisticians)), and the fact that $Var(\frac{S_U(\vC)}{\sqrt{n}}) \to \mbox{sd}^2$ as shown in the proof of Theorem \ref{thm:main}, we have
\bqan \label{berry_essen_bound}
\sup_{\vC}\sup_{t}\left|P\left(\frac{S_U(\vC)}{\sqrt{n}\mbox{sd}} \leq t \right) - \Phi\left(t\right)\right| \leq  9\sup_{\vC}\frac{\sum_{i=1}^{r_n}E\left|U_i(\vC)\right|^3}{\left[\sum_{i=1}^{r_n}\mbox{Var}(U_i(\vC))\right]^{3/2}} = O\left(\frac{1}{\sqrt{r_n}}\right) = o(1).
\eqan
Let $t^* = t - \frac{S_V(\vC)}{\sqrt{n}\mbox{sd}}$, then
\bqan 
&& \sup_{\vC}\left|P\left(\frac{S_U(\vC)}{\sqrt{n}\mbox{sd}} \leq t - \frac{S_V(\vC)}{\sqrt{n}\mbox{sd}}, \left|\frac{S_V(\vC)}{\sqrt{n}\mbox{sd}}\right| \leq \epsilon \right) - \Phi\left(t\right)\right| \nonumber\\
&=& \sup_{\vC}\Bigg|P\left(\frac{S_U(\vC)}{\sqrt{n}\mbox{sd}} \leq t^*, \left|\frac{S_V(\vC)}{\sqrt{n}\mbox{sd}}\right| \leq \epsilon \right) - P\left(\frac{S_U(\vC)}{\sqrt{n}\mbox{sd}} \leq t^* \right) \nonumber\\
&&\ \ \ \ \ \ \ \ \ \ \ \ + P\left(\frac{S_U(\vC)}{\sqrt{n}\mbox{sd}} \leq t^* \right) - \Phi(t^*) + \Phi(t^*) - \Phi(t)\Bigg| \nonumber\\
& \leq & \sup_{\vC}\Bigg|P\left(\frac{S_U(\vC)}{\sqrt{n}\mbox{sd}} \leq t^*, \left|\frac{S_V(\vC)}{\sqrt{n}\mbox{sd}}\right| \leq \epsilon \right) - P\left(\frac{S_U(\vC)}{\sqrt{n}\mbox{sd}} \leq t^* \right)\Bigg| \nonumber\\
&&\ \ \ \ \ \ \ \ \ \ \ \ + \sup_{\vC}\Bigg|P\left(\frac{S_U(\vC)}{\sqrt{n}\mbox{sd}} \leq t^* \right) - \Phi(t^*)\Bigg| + \Bigg|\Phi(t^*) - \Phi(t)\Bigg|.\label{Su_sum}
\eqan
The first term in (\ref{Su_sum}) goes to zero by continuity of measures, since by Lemma \ref{lemma:Sv} $P\left(\left|\frac{S_V(\vC)}{\sqrt{n}\mbox{sd}}\right| \leq \epsilon\right) \to 1$. The second term in (\ref{Su_sum}) goes to zero by (\ref{berry_essen_bound}), and the third term goes to zero by the continuity of $\Phi(.)$.

\end{proof}

\begin{lemma} \label{lemma:Fn} 
Let $X_1,\ldots,X_n$ be iid[$F$], and let $\hat{F}_n(x)$ be the corresponding empirical distribution function. Then, for any constant $c$, $$sup_{x_i,x_j}\left\{|F(x_i) - F(x_j)|I\left[|\hat{F}(x_i) - \hat{F}(x_j)| \leq \frac{c}{n}\right]\right\} = O_p\left(\frac{1}{\sqrt{n}}\right).$$
\end{lemma}
\begin{proof}  By the Dvoretzky, Kiefer and Wolfowitz (1956) theorem, we have that 
$\forall \epsilon \geq 0,$
$$
 P\left(\sup_{x}|\hat{F}_n(x) - F(x)| \geq \epsilon\right) \leq Ce^{-2n\epsilon^2}.
$$
Therefore, $|\hat{F}(x) - F(x)| = O_p\left(\frac{1}{\sqrt{n}}\right)$ uniformly on $x$. Hence, writing
\begin{eqnarray}
|F(x_i) - F(x_j)| &=& |F(x_i) -\hat{F}_n(x_i) + \hat{F}_n(x_i)- F(x_j) + \hat{F}_n(x_j) - \hat{F}_n(x_j)|, \nonumber 
\end{eqnarray}
it follows that $sup_{x_i,x_j}\left\{|F(x_i) - F(x_j)|I\left[|\hat{F}(x_i) - \hat{F}(x_j)| \leq c/n\right]\right\}$ is less than or equal to
\bqa
 &  & sup_{x_i,x_j}\left\{|F(x_i) -\hat{F}_n(x_i)| + |\hat{F}_n(x_j) - F(x_j)|\right\} 
 \nonumber \\
 &&+
  sup_{x_i,x_j}\left\{|\hat{F}_n(x_i)- \hat{F}_n(x_j)|\right\}I\left[|\hat{F}_n(x_i) - \hat{F}_n(x_j)| \leq \frac{c}{n}\right] \nonumber \\
 &=&  O_p\left(\frac{1}{\sqrt{n}}\right) + O_p\left(\frac{1}{\sqrt{n}}\right) + O_p\left(\frac{1}{n}\right).  
\eqa
This completes the proof of the lemma.\end{proof}

\begin{lemma} \label{lemma:Op} 
With $W_i$ be defined in (\ref{def.Wi}), and any Lipschitz continuous function $g(x)$,  
\begin{equation} \label{eqn:Op}
\frac{1}{p}\sum_{j=1}^{n}g(x_{2j})I(j \in W_i) - g(x_{2i}) = O_p\left(\frac{1}{\sqrt{n}}\right), \nonumber
\end{equation}
uniformly in $i=1,\ldots,n$.
\end{lemma}
\begin{proof} First note that by the Lipschitz continuity and the Mean Value Theorem we have
\begin{eqnarray} \label{eqn:Op2}
&& |g(x_{2j}) - g(x_{2i})| \leq M|x_{2j} - x_{2i}| \leq M|F_{X_2}(x_{2j}) - 
F_{X_2}(x_{2i})|/f_{X_2}(\tilde{x}_{ij}), \nonumber
\end{eqnarray} 
for some constant $M$,
where $\tilde x_{ij}$ is between $x_{2j}$ and $x_{2i}$. Thus, 
\bqa
&&\hspace{-4mm}\left |\frac{1}{p}\sum_{j=1}^{n}g(x_{2j})I(j \in W_i) - g(x_{2i})\right| 
\leq \frac{1}{p}\sum_{j=1}^{n}|g(x_{2j}) - g(x_{2i})|I\left[|\hat{F}_{X_2}(x_{2i}) - \hat{F}_{X_2}(x_{2j})| \leq \frac{p-1}{2n}\right] \nonumber \\
  & &\leq  \frac{M}{p}\sum_{j=1}^{n}\frac{|F_{X_2}(x_{2j}) - F_{X_2}(x_{2i})|}{f_{X_2}(\tilde{x}_{ij})}I\left[|\hat{F}_{X_2}(x_{2i}) - \hat{F}_{X_2}(x_{2j})| \leq \frac{p-1}{2n}\right] = O_p\left(\frac{1}{\sqrt{n}}\right), \nonumber
\eqa
where the last equality follows from Lemma \ref{lemma:Fn} and the assumption that $f_{X_2}$ remains bounded away from zero.
\end{proof}

\makeatletter
\renewcommand\@biblabel[1]{}
\makeatother
\renewcommand{\baselinestretch}{1.00}

\end{document}